\documentclass{article}

\usepackage{microtype}
\usepackage{graphicx}
\usepackage{subcaption}
\usepackage{booktabs} 
\usepackage{tikz}
\usetikzlibrary{quantikz2,shapes,arrows,positioning,calc}

\usepackage{hyperref}

\usepackage{algorithm}
\usepackage{algorithmic}
\usepackage{footnote}
\usepackage[preprint]{icml2026}

\usepackage{amsmath}
\usepackage{amssymb}
\usepackage{mathtools}
\usepackage{amsthm}
\usepackage{braket} % for Dirac notation
\usepackage{xcolor}
\DeclareMathOperator*{\argmin}{argmin}
\usepackage{subcaption}
\usepackage[capitalize,noabbrev]{cleveref}

\theoremstyle{plain}
\newtheorem{theorem}{Theorem}[section]

\newtheorem{lemma}[theorem]{Lemma}
\newtheorem{corollary}[theorem]{Corollary}
\theoremstyle{definition}
\newtheorem{definition}[theorem]{Definition}

\theoremstyle{remark}

\usepackage{tabularx}
\usepackage{multirow}
\usepackage[textsize=tiny]{todonotes}

\usepackage{tablefootnote}
\usepackage{xspace}
\newcommand{\nqtabs}{\textbf{\textbf{NNQA\xspace}}}
\newcommand{\nq}{\textbf{\texttt{NNQA\xspace}}}

\icmltitlerunning{NNQA: Neural-Native Quantum Arithmetic for End-to-End Polynomial Synthesis}

\begin{document}

\twocolumn[
  \icmltitle{NNQA: Neural-Native Quantum Arithmetic for End-to-End Polynomial Synthesis}

  % Author information will be automatically removed for blind review
\icmlsetsymbol{equal}{*}

\begin{icmlauthorlist}
  \icmlauthor{Ziqing Guo}{ttu,comp}
  \icmlauthor{Jie Li}{ttu}
  \icmlauthor{Yong Chen}{ttu}
  \icmlauthor{Ziwen Pan}{ttu}
\end{icmlauthorlist}

\icmlaffiliation{ttu}{Department of Computer Science, Texas Tech University, TX, USA}
\icmlaffiliation{comp}{Lawrence Berkeley National Laboratory, CA, USA}

  \icmlcorrespondingauthor{Ziwen Pan}{ziwen.pan@ttu.edu}

  \icmlkeywords{Quantum Neural Networks, Neural Operators, Quantum Computing, Neural Quantum States, Variational Quantum Algorithms}

  \vskip 0.3in
]

\printAffiliationsAndNotice{}

\begin{abstract}
Hybrid classical–quantum learning is often bottlenecked by communication overhead and approximation error from generic variational ansatzes. In this study, we introduce Neural-Native Quantum Arithmetic (NNQA), which compiles classically learned nonlinear representations into precise quantum arithmetic composed of native unitary blocks. Theoretically, we prove that the universal approximation of quantum polynomial arithmetic can be realized by transforming a classical neural network into a quantum circuit, with the resulting error arising solely from measurement shot noise, thereby extending classical operator-level estimation guarantees into the quantum regime.  Empirical validation on IBM Quantum Heron3 and IonQ Forte processors shows performance limited primarily by device noise without variational fine-tuning: we achieve over 99.5\% accuracy for polynomials up to degree 35 and demonstrate scalability on IonQ hardware up to 36 qubits and circuit depths of 70, reaching a negligible RMSE of 0.005. Overall, NNQA establishes a new paradigm of synthesizing quantum arithmetic for native quantum computation.
\end{abstract}

%======================================================================
\section{Introduction}
\label{sec:intro}
%======================================================================

Quantum processors are increasingly deployed as accelerators within high-performance computing (HPC) workflows, whereas machine learning is simultaneously permeating the quantum stack from compilation to error mitigation~\cite{alexeev2025aiqc, mohseni2024build}. This convergence makes the classical--quantum interface a first-order systems concern: when a quantum device is accessed via a host environment, iterative hybrid algorithms are often bottlenecked by latency and bandwidth at the communication boundary rather than by quantum execution times~\cite{mcclean2018barren, zhang2024exponential}.

In hybrid learning techniques, the dominant framework conceptualizes parameterized quantum circuits as models that are trainable through optimization within a classical loop~\cite{mcclean2016theory,cerezo2021variational}. Although this train-a-circuit methodology is effective for representation learning, it is not well suited for numerical computation and operator evaluation. First, the accuracy is intertwined with loop optimization, which necessitates repeated hardware queries, thereby increasing the communication overhead. Second, arithmetic operations are indirectly represented through generic rotation gates, which introduce approximation errors even before considering hardware noise. In contrast, manually designed arithmetic circuits~\cite{vedral1996quantum,draper2000addition, ruiz2017quantum} can implement exact primitives but are inadequate for scaling complex nonlinear operators required in scientific computing.

Recent research has investigated classical models for addressing quantum problems, including Neural Quantum States (NQS)~\cite{carleo2017solving,zhang2023tqs}, as well as reinforcement~\cite{kuo2021quantum} and deep learning-based compilations~\cite{bondesan2021hintons}. 
Although these methods enhance state representation or circuit fine-tuning, they do not offer a mechanism to execute learned operators on quantum hardware. Specifically, the classical model learns the quantum wavefunction representation but lacks a direct pathway to convert it into a hardware-executable quantum circuit without a variational convergence. We emphasize that while diffusion models have been effectively employed to synthesize quantum circuits for predefined unitaries, the methods focus on gate-level compilation rather than end-to-end learning of arithmetic representations from classical data \cite{furrutter2024quantum}.

To address the communication bottleneck gap, we reconceptualized quantum arithmetic synthesis as a problem of with direct compilation from a machine learning model. \nq\ provides a compile-then-execute approach that maintains the arithmetic inductive bias, eliminates variational fine-tuning, and ensures that circuit synthesis introduces zero approximation error. Our method replicates the classical outcomes, with the sole deviation attributable to the quantum-native measurement shot noise.

We established a rigorous theoretical connection between classical model learning and quantum computing. We demonstrate that the expectation values of our native arithmetic basis conform to statistically secure results, thereby extending the universal approximation theorem to the quantum simulation domain. Empirically, we validated our approach on both superconducting and trapped-ion QPUs, achieving a high-fidelity polynomial evaluation in the noise intermediate-scale quantum (NISQ) era \cite{preskill2018quantum}. These features enable future applications for direct quantum simulation of nonlinear differential equations, native quantum models, and physical quantum computation.
\textbf{Our contributions are summarized as follows}:
\begin{itemize}
    \item We formulate a theoretical framework bridging classical neural approximation with quantum arithmetic, proving universality under a native arithmetic basis.
    \item We derive a closed-form mapping that converts neural network polynomial coefficients into quantum rotation angles, enabling deterministic circuit construction without hybrid optimization.
    \item We identify and characterize the error model, showing that the ideal polynomial evaluation on expectation values is limited solely by shot noise, independent of polynomial degree.
    \item We validate \nq\ on the state-of-the-art NISQ devices, demonstrating superior accuracy and efficiency compared to variational quantum method baselines for nonlinear approximation tasks.
\end{itemize}

% The rest of the paper is organized as follows: \cref{sec:related} reviews background and related work; \cref{sec:methodology} presents the NNQA framework; \cref{sec:theory} provides theoretical analysis; \cref{sec:experiments} details experimental validation; \cref{sec:conclusion} concludes our .
%======================================================================
\section{Background and Related Work}
\label{sec:related}
%======================================================================

%----------------------------------------------------------------------
\subsection{Universal Approximation Theory}
\label{sec:universal}
%----------------------------------------------------------------------

The expressiveness of function approximators is governed by universal approximation theorems (UATs), which establish theoretical completeness for both classical \cite{lu2021learning} and quantum architectures \cite{wangquanonet}.

\paragraph{Classical Models.}
The foundational UAT for neural networks~\cite{hornik1989multilayer,cybenko1989approximation} states that feedforward networks with nonlinear activations can approximate any continuous function $f: [0,1]^n \to \mathbb{R}$ to arbitrary precision $\epsilon$. For polynomial networks, the Weierstrass approximation theorem~\cite{stone1948generalized} guarantees that polynomials can uniformly approximate continuous functions on compact domains. This ensures that the polynomial representations learned by the neural networks are theoretically capable of universal function approximation.

\paragraph{Quantum Circuits.}
Analogous results have been obtained for parameterized quantum circuits (PQCs) \cite{younis2022quantum,benedetti2019parameterized}. The Solovay-Kitaev theorem \cite{dawson2005solovay} guarantees that universal gate sets can approximate any unitary $U \in \mathrm{SU}(2^n)$ efficiently. Data re-uploading models~\cite{schuld2021effect} have been shown to be universal function approximators because their Fourier series representation can fit any target function given sufficient depth. However, realizing this universality in practice is hindered by trainability issues, such as barren plateaus~\cite{mcclean2018barren}, where gradients vanish exponentially with the system size. Specifically,  
\begin{theorem}
\label{thm:barren}
For a randomly initialized PQC of depth $D$ on $n$ qubits, the variance of the gradients satisfies $\mathrm{Var}[\partial \langle O \rangle / \partial \theta_j] \leq c/2^n$ for constant $c$, making gradient-based optimization exponentially difficult.
\end{theorem}

\subsection{Variational and Neural Quantum Approaches}
\label{sec:vqa}

Variational Quantum Algorithms (VQAs) constitute a leading paradigm for near-term quantum computing, utilizing parameterized models to approximate quantum states and operators.

\paragraph{Variational Quantum Eigensolver (VQE).}
This approach \cite{peruzzo2014variational} optimizes a parameterized quantum circuit to minimize the energy functional $E(\boldsymbol{\theta}) = \langle \psi(\boldsymbol{\theta}) | H | \psi(\boldsymbol{\theta}) \rangle$. By the variational principle, $E(\boldsymbol{\theta}) \geq E_0$ for all $\boldsymbol{\theta}$, ensuring that the gradient descent on the quantum-measured objective converges toward the ground state energy $E_0$. Recent advances have focused on improving ansatz expressibility~\cite{du2022efficient,wu2021expressivity} and developing gradient-free or error-aware optimization strategies to enhance the convergence in noisy environments. Similarly, such hybrid approaches have been extended to generalized supervised learning tasks~\cite{tacchino2020quantum,alchieri2021introduction}, where parameterized circuits are trained to classify data or approximate functions.

\paragraph{Neural Quantum States (NQS).}
Introduced by Carleo and Troyer~\cite{carleo2017solving}, the NQS framework employs artificial neural networks to represent the complex amplitudes of a many-body quantum wavefunction, $\Psi_\theta(\mathbf{s})$. Using architectures such as Restricted Boltzmann Machines (RBMs) or deep convolutional networks, NQS compactly approximates highly entangled states that are challenging for traditional tensor network methods \cite{biamonte2017tensor}. Note that training is typically performed using Variational Monte Carlo (VMC) sampling to optimize the network parameters according to the variational principle, allowing for the simulation of ground-state properties and unitary dynamics on classical hardware.

%----------------------------------------------------------------------
\subsection{Quantum Arithmetic Primitives}
\label{sec:quantum-arithmetic}
%----------------------------------------------------------------------

A robust mapping from classical data to quantum observables is required to enable the direct quantum execution of classically learned functions.

\paragraph{Classical Data to Hilbert Space.}
Standard encodings (amplitude, angle, basis) map data $x$ to state $\ket{\psi(x)}$ \cite{rath2024quantum, weigold2020data}, but extracting results often requires full-state tomography or deep arithmetic circuits (e.g., QFT-based adders~\cite{draper2000addition}), which are prohibitive for NISQ devices. This can be mitigated with classical preprocessing techniques, revealing that using Pauli measurement-based tomography leads to faster classical data recovery \cite{huang2021efficient, huang2020predicting}. The particular angle encoding technique is Expectation-Value Encoding (EVEN), which maps classical data directly to the statistical properties of observables that depend on the Monte Carlo simulation \cite{balewski2025ehands}. The EVEN protocol is defined as follows:

\begin{definition}
\label{def:even}
For input $x \in [-1, 1]$, we encode $x$ into a qubit state via $\theta = \arccos(x)$.
\begin{equation}
    \ket{\psi(x)} = R_y(\theta)\ket{0} = \sqrt{\frac{1+x}{2}}\ket{0} + \sqrt{\frac{1-x}{2}}\ket{1}
\end{equation}
This guarantees that the Pauli-$Z$ expectation value exactly recovers the input: $\langle Z \rangle = x$.
\end{definition}
This formulation transforms the function evaluation into the construction of circuits that manipulate the expectation values. The generalized quantum protocol for performing multiplication and weighted summation defines an exact arithmetic basis for this purpose \cite{guo2025vectorized}:

\begin{definition}
\label{def:1}
For independent qubits encoding $x_0, x_1$, the primitive $U_{\mathrm{mult}} = \mathrm{CNOT}_{0,1} \cdot (I \otimes R_z(\pi/2))$ yields:
\begin{equation}
    \langle I \otimes Z \rangle_{U_{\mathrm{mult}}} = \langle Z_0 \rangle \cdot \langle Z_1 \rangle = x_0 \cdot x_1
\end{equation}
\end{definition}

\begin{definition}
\label{def:2}
For qubits encoding $x_0, x_1$ and weight $w \in [0,1]$, the unitary $U_{\mathrm{sum}}(\alpha)$ with $\alpha = \arccos(1-2w)$ yields:
\begin{equation}
    \langle Z \otimes I \rangle_{U_{\mathrm{sum}}} = w \cdot x_0 + (1-w) \cdot x_1
\end{equation}
\end{definition}

These primitives form a complete basis for polynomial evaluation, which can be summarized by the following theorem:

\begin{theorem}
\label{thm:poly-exact}
Any degree-$d$ polynomial $P_d(x)$ with $|P_d(x)| \leq 1$ can be computed exactly on the expectation values using $d+1$ qubits and depth $3d+1$. The final result is:
\begin{equation}
    \langle Z_{\mathrm{out}} \rangle = P_d(x) + \epsilon_{\mathrm{shot}}
\end{equation}
where $\epsilon_{\mathrm{shot}} \sim O(1/\sqrt{N})$ is the error arising from finite sampling.
\end{theorem}

We emphasize that because each classical input can be produced from expectation value that encoded by angle rotation of each parameterized gate, this enables a clean separation between 
classical training and quantum execution. We will discuss the  polynomial coefficients map to rotation 
angles in closed form without acquiring full state tomography in \cref{sec:nqaa-theorem}. Alternatively, Quantum Signal Processing (QSP) offers a paradigm for implementing polynomial transformations of input signals with optimal query complexity by interleaving signal-rotation operators $R_x(\theta)$ with signal-processing rotations $e^{i\phi Z}$~\cite{low2017optimal}. In addition, Quantum Singular Value Transformation (QSVT)~\cite{gilyen2019quantum} provides a solution to matrix arithmetics and Hamiltonian simulation. However, the standard QSP sequence relies on phase factors $\{\phi_k\}$ determined via classical optimization algorithms that become numerically unstable at high degrees~\cite{chao2020finding}, although recent geometric methods have improved stability~\cite{dong2021efficient}. The main differences between the quantum computational paradigms are summarized in Table~\ref{tab:comparison1}. In the following sections, we build upon the quantum arithmetic primitives to construct the \nq\ framework.

\begin{table}[htbp]
\centering
\caption{Comparison of \textbf{NNQA (this work)} with representative quantum learning and arithmetic frameworks.}
\label{tab:comparison1}
\small
\begin{tabular}{lcccc}
\toprule
Property & \textbf{NNQA} & NQS & VQE & QSP \\
\midrule
Classical Training      & \checkmark & \checkmark & \texttimes & \checkmark \\
Exact Quantum Exec.     & \checkmark & \texttimes & \texttimes & \checkmark \\
No Barren Plateaus      & \checkmark & \checkmark & \texttimes & \checkmark \\
No Q-C Loop             & \checkmark & \texttimes & \texttimes & \checkmark \\
NISQ Compatible         & \checkmark & \texttimes & \checkmark & \texttimes \\
Limited Circuit Depth   & \checkmark & \texttimes & \texttimes & \texttimes \\
\bottomrule
\end{tabular}

\vspace{0.5em}
\begin{minipage}{\linewidth}
\footnotesize \textit{Note:} Exact quantum execution refers to the ability to compute functions without approximation errors beyond shot noise. The Q-C Loop indicates a post-processed quantum measurement result combined with classical gradient optimization.
\end{minipage}
\end{table}
%======================================================================
\section{NQAA Universal Approximation Theorem}
\label{sec:methodology}
%======================================================================

% We present Neural-Native Quantum Arithmetic (NQAA), a framework that bridges classical neural network training with exact quantum circuit execution.
% Our approach follows the three principles (i)~establishing theoretical foundations via a quantum extension of the UAT; (ii)~designing polynomial neural network architectures compatible with quantum arithmetic primitives; and (iii)~implementing NISQ-compatible circuits using quantum arithemetics.

%----------------------------------------------------------------------
\subsection{Theoretical Results}
\label{sec:nqaa-theorem}
%----------------------------------------------------------------------

First, we establish that polynomial operators learned by neural networks can be exactly represented on quantum hardware, mapping the classical universal approximation theory to the quantum representation. Although the Pauli operator group is sufficient to fully reconstruct the classical tomography in the Hilbert computation space \cite{xiao2025quantum,wangquanonet}, the tunable hybrid layer depends on quantum optimization, such as parameter shift rules \cite{wierichs2022general} and iterative classical–quantum training loops~\cite{resch2021introductory, smaldone2025hybrid}. The universal quantum approximation is not unique. The first challenge is that quantum operations are linear and unitary, whereas classical functions can be nonlinear. Second, quantum measurements introduce stochasticity, which complicates the exact function representation. In other words, recovering the nonlinear function from quantum linear operations with minimum stochastic measurements is the key to establishing the universal approximation theorem in the quantum domain. 

Hence, we provide that our theoretical results are based on discrete Monte Carlo simulations using quantum arithmetic circuits (\cref{sec:quantum-arithmetic}), where the nonlinear function can be recovered by the Pauli-$Z$-basis computational measurement expectation value. This leads to the classical Universal Approximation Theorem (UAT) with exact quantum execution because the angle-encoded classical data are derived from simulated shots, circumventing the iterative optimization loops from classical and quantum data communication and barren plateaus from training complex Hilbert space that are inherent to variational methods. The benefits become clearer when running on a future fault-tolerant quantum computer (FTQC) because the Monte Carlo shots are only affected by the number of qubits, which increases logarithmically with the size of the problem. Thus, we present the theorem for \nq\ as follows:

\begin{theorem}
\label{thm:nqaa-uat}
Let $\mathcal{F} \in C([-1,1])$ be a continuous function defined on $[-1,1]$, and for any $\epsilon > 0$, there exists a polynomial $P_d(x) = \sum_{k=0}^{d} a_k x^k$ and quantum circuit $U_P$ constructed from the basis $\mathcal{B} = \{U_{\mathrm{mult}}, U_{\mathrm{sum}}\}$ such that the circuit output expectation is $\langle Z_{\mathrm{out}} \rangle = P_d(x)$. The approximation error of the estimator $\hat{P}_d$ (from $N$ shots) is decomposed as
\begin{equation}
    |\mathcal{F}(x) - \hat{P}_d(x)| \leq \underbrace{\|\mathcal{F} - P_d\|_\infty}_{\epsilon_{\mathrm{classical}}} + \underbrace{O(1/\sqrt{N})}_{\epsilon_{\mathrm{shot}}}
\end{equation}
where $\epsilon_{\mathrm{classical}}$ is controlled by the degree $d$, and $\epsilon_{\mathrm{shot}}$ depends only on $N$.
\end{theorem}

\begin{proof}
By the Weierstrass approximation theorem~\cite{stone1948generalized}, for any $\epsilon > 0$, there exists a polynomial $P_d$ such that $\|\mathcal{F} - P_d\|_\infty < \epsilon_{\mathrm{classical}}$. We assume that the coefficients are normalized such that $\|P_d\|_\infty \le 1$.

We then construct the quantum circuit $U_P$ using quantum arithmetic as follows: First, we define the state $\ket{\psi_k}$ such that $\langle Z_k \rangle = x^k$ using recursive multiplication:
\begin{align}
    \ket{\psi_k} &= U_{\mathrm{mult}}^{(k-1,k)} \ket{\psi_{k-1}} \otimes \ket{\psi(x)}, \\ \langle Z_k \rangle &= \langle Z_{k-1} \rangle x = x^k.
\end{align}
The terms are aggregated via backward recursion. Let $S_d = a_d x^d$ and for $k=d \dots 1$:
\begin{align}
    S_{k-1} &= w_{k-1} (a_{k-1} x^{k-1}) + (1-w_{k-1}) S_k, \\ w_k &= \frac{|a_k|}{\sum_{j=k}^{d} |a_j|}
    \label{eq:weight-mapping}
\end{align}
Applying $U_{\mathrm{sum}}(\alpha_k)$ with $\alpha_k = \arccos(1-2w_k)$ implements this recursion, yielding $\langle Z_{\mathrm{out}} \rangle = S_0 = P_d(x)$. We defer the aggregation proof to \cref{app:arithmetic_proofs}.
Finally, we consider the estimator $\hat{P}_d$ obtained from $N$ shots. It has expectation $P_d(x)$ and variance $\mathrm{Var}[\hat{P}_d] \le 1/N$. According to Hoeffding's inequality, the probability of deviation decays exponentially: 
\begin{equation}
   \Pr[|\hat{P}_d - P_d| > \epsilon] \leq 2\exp(-2N\epsilon^2).
\end{equation}
Combining the approximation error from Weierstrass and the estimation error from sampling via the triangle inequality
\begin{equation}
    |\mathcal{F}(x) - \hat{P}_d(x)| \leq |\mathcal{F}(x) - P_d(x)| + |P_d(x) - \hat{P}_d(x)|
\end{equation}
yields the error decomposition, as stated.
\end{proof}
Because we leverage classical training for coefficient learning, the optimization error $\epsilon_{\mathrm{opt}}$ inherent to variational methods is removed. The polynomial approximation error $\epsilon_{\mathrm{classical}}$ is controlled by $d$, whereas the shot noise $\epsilon_{\mathrm{shot}}$ depends only on $N$. This decoupling allows for the independent tuning of accuracy using classical and quantum resources. Therefore, we propose the following corollary:
\begin{corollary}
\label{thm:coeff-to-weight}
The mapping from the learned coefficients $\{a_k\}$ to the quantum circuit parameters $\{\alpha_k\}$ is given by a deterministic closed-form expression. For coefficients normalized by $\mathcal{C} \ge \|\sum a_k x^k\|_\infty$, the rotation angles are derived as
\begin{align}
    \alpha_k &= \arccos\left(1 - \frac{2|\tilde{a}_k|}{\sum_{j=k}^{d} |\tilde{a}_j|}\right), \quad \text{where } \tilde{a}_k = \frac{a_k}{\mathcal{C}}, \\
    \mathcal{C} &= \max_{x \in [-1,1]} \left| \sum_{k=0}^d a_k x^k \right|.
\end{align}
This transformation requires $O(d)$ classical operations and introduces no optimization errors.
\end{corollary}
Here, the measurement of the output qubit yields
\begin{equation}
    \label{eq:final-estimator}
    \hat{P}_d(x) = \mathcal{C} \cdot \frac{n_0 - n_1}{N},
\end{equation}
with error bound statistics
\begin{equation}
    n_1 + n_0 = N, \quad p = \frac{n_1}{N}, \quad \sigma(p) = \sqrt{\frac{p(1-p)}{N^3}},
    \label{eq:shots_noise}
\end{equation}
where $\mathcal{C}$ is the normalization constant from
\begin{equation}
    \tilde{a}_k = \frac{a_k}{\mathcal{C}}, \quad \mathcal{C} = \max_{x \in [-1,1]} \left| \sum_{k=0}^{d} a_k x^k \right| + \epsilon
    \label{eq:normalization}
\end{equation}
derived from \cref{thm:coeff-to-weight}. \cref{fig:shots} shows the shots noise error bound.
\begin{figure}[htbp]
    \centering
    \includegraphics[width=0.98\columnwidth]{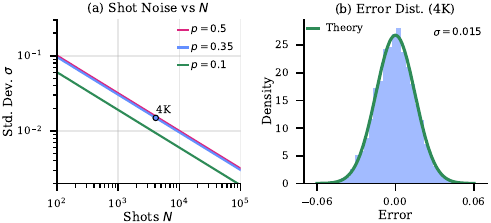}
    \caption{Error scaling of \nq\ for quantum simulation, where the noise is controlled by degree $d$ and shots. The error is decreasing as $O(1/\sqrt{N})$.}
    \label{fig:shots}
\end{figure}

Specifically, \cref{thm:coeff-to-weight} allows the classically trained network to be directly transformed into the angle-encoded paradigm, which eliminates the $O(pT)$ classical-quantum communication bound required by variational methods with $p$ parameters and $T$ iterations.
This analytical mapping fundamentally alters the error landscape compared with variational methods.  In the Variational Quantum Eigensolver (VQE) and  Quantum Neural Networks (QNNs), the total error is typically decomposed as $\epsilon_{\mathrm{opt}} + \epsilon_{\mathrm{ansatz}} + \epsilon_{\mathrm{shot}}$, where $\epsilon_{\mathrm{opt}}$ arises from non-convex quantum optimization landscapes (e.g., barren plateaus) and $\epsilon_{\mathrm{ansatz}}$ arises from limited circuit expressibility. Here in \nq, because the optimization involves convex classical training ($\epsilon_{\mathrm{opt}}^{\text{quantum}} = 0$) and the circuit construction exactly realizes the polynomial ($\epsilon_{\mathrm{ansatz}} = 0$), the error budget is reduced to $\epsilon_{\mathrm{classical}} + \epsilon_{\mathrm{shot}}$. This decoupling allows us to leverage mature classical deep learning optimizers to minimize $\epsilon_{\mathrm{classical}}$ independent of quantum hardware constraints, while $\epsilon_{\mathrm{shot}}$ remains a function of measurement resources. An ideal resource analysis is provided in \cref{sec:resource}.

%======================================================================
\section{\nqtabs\ Computational Architecture}
\label{sec:architecture}
%======================================================================
\subsection{Model Overview}
\label{sec:model-overview}
We describe the computational architecture that implements \cref{thm:nqaa-uat}. The system comprises three sequential phases: classical training, deterministic compilation, and quantum execution, as shown in \cref{fig:architecture}. We also refer to the computation algorithm of \nq\ in \cref{alg:nnqa}.
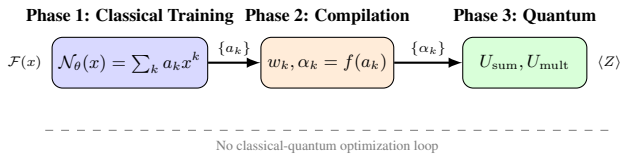
\begin{figure}[htbp]
\centering
\begin{tikzpicture}[scale=0.75, transform shape, >=latex]
    % Phase 1: Classical Training
    \node[draw, rounded corners, minimum width=2.2cm, minimum height=0.9cm, fill=blue!15] (NN) at (-3.5, 0) {\footnotesize $\mathcal{N}_\theta(x) = \sum_k a_k x^k$};
    \node[above=0.1cm of NN] {\footnotesize \textbf{Phase 1: Classical Training}};
    
    % Phase 2: Compilation
    \node[draw, rounded corners, minimum width=2.2cm, minimum height=0.9cm, fill=orange!15] (MAP) at (0, 0) {\footnotesize $w_k, \alpha_k = f(a_k)$};
    \node[above=0.1cm of MAP] {\footnotesize \textbf{Phase 2: Compilation}};
    
    % Phase 3: Quantum Execution
    \node[draw, rounded corners, minimum width=2.2cm, minimum height=0.9cm, fill=green!15] (QC) at (3.5, 0) {\footnotesize $U_{\mathrm{sum}}, U_{\mathrm{mult}}$};
    \node[above=0.1cm of QC] {\footnotesize \textbf{Phase 3: Quantum}};
    
    % Arrows
    \draw[->, thick] (NN) -- node[above, font=\scriptsize] {$\{a_k\}$} (MAP);
    \draw[->, thick] (MAP) -- node[above, font=\scriptsize] {$\{\alpha_k\}$} (QC);
    
    % Input/Output
    \node[left=0.05cm of NN, font=\scriptsize] {$\mathcal{F}(x)$};
    \node[right=0.05cm of QC, font=\scriptsize] {$\langle Z \rangle$};
    
    % Bottom: No Q-C loop indicator
    \draw[dashed, gray] (-5, -1.2) -- (5, -1.2);
    \node[font=\scriptsize, gray] at (0, -1.5) {No classical-quantum optimization loop};
\end{tikzpicture}
\caption{Classical training produces polynomial coefficients $\{a_k\}$, which are compiled to rotation angles $\{\alpha_k\}$ via \cref{thm:coeff-to-weight}, then composed into quantum polynomial circuits.}
\label{fig:architecture}
\end{figure}

\begin{algorithm}[htbp]
\caption{Neural-Native Quantum Arithmetic}
\label{alg:nnqa}
\small
\begin{algorithmic}[1]
\REQUIRE Target function $\mathcal{F}: [-1,1] \to \mathbb{R}$, degree $d$
\ENSURE Compiled circuit parameters $(\{\alpha_k\}, \mathcal{C})$

\STATE \textbf{Phase 1: Classical Training} \hfill $\triangleright$ \cref{sec:architecture}
\STATE Sample training data $\{(x_i, \mathcal{F}(x_i))\}_{i=1}^{M}$
\STATE $\{a_k\}_{k=0}^{d} \gets \texttt{TrainPolynomialNN}(\{x_i, \mathcal{F}(x_i)\}, d)$ \hfill $\triangleright$ \cref{eq:nn-loss}
\STATE $\mathcal{C} \gets \max_{x \in [-1,1]} \left|\sum_{k=0}^{d} a_k x^k\right| + \epsilon$ \hfill $\triangleright$ \cref{eq:normalization}
\STATE $\tilde{a}_k \gets a_k / \mathcal{C}$ for all $k$

\STATE \textbf{Phase 2: Parameter Mapping} \hfill $\triangleright$ 
\FOR{$k = d$ down to $0$}
    \STATE $w_k \gets |\tilde{a}_k| \,/\, \sum_{j=k}^{d} |\tilde{a}_j|$ \hfill $\triangleright$ \cref{thm:coeff-to-weight}
    \STATE $\alpha_k \gets \arccos(1 - 2w_k)$ \hfill $\triangleright$ 
\ENDFOR

\STATE \textbf{Phase 3: Quantum Evaluation} \hfill $\triangleright$ 
\STATE \textbf{function} \textsc{Evaluate}($x$, shots $N$):
\STATE \quad $\theta_x \gets \arccos(x)$ \hfill $\triangleright$ \cref{def:even}
\STATE \quad $\texttt{qc} \gets \texttt{BuildCircuit}(\theta_x, \{\alpha_k\}, d)$ \hfill $\triangleright$ 
\STATE \quad $(n_0, n_1) \gets \texttt{Execute}(\texttt{qc}, N)$
\STATE \quad \textbf{return} $\mathcal{C} \cdot (n_0 - n_1) / N$ \hfill $\triangleright$ \cref{eq:final-estimator}
\end{algorithmic}
\end{algorithm}

A polynomial neural network $\mathcal{N}_\theta: \mathbb{R} \to \mathbb{R}$ learns coefficients $\{a_k(\theta)\}_{k=0}^d$ by minimizing the mean squared error on training samples $\{(x_i, \mathcal{F}(x_i))\}_{i=1}^M$:
\begin{equation}
    \label{eq:nn-loss}
    \theta^* = \argmin_\theta \frac{1}{M} \sum_{i=1}^M \left( \sum_{k=0}^d a_k(\theta) x_i^k - \mathcal{F}(x_i) \right)^2
\end{equation}
Training employs standard backpropagation with exact gradients, circumventing the stochastic gradient estimates and barren plateaus associated with variational quantum training. The learned coefficients $a_k := a_k(\theta^*)$ are subsequently normalized to adhere to the quantum encoding constraint $\|P_d\|_\infty \leq 1$. In this context, the normalization constant $\mathcal{C}$ is documented for post-measurement rescaling. For input encoding, the aggregation of terms through sequential weighted sums and rotation angles is delineated in \cref{eq:weight-mapping}. In instances where the coefficients $\tilde{a}_k$ are negative, a sign inversion is performed by incorporating an $X$ gate on the corresponding input qubit. This compilation necessitates $O(d)$ classical arithmetic operations because the quantum arithmetic operators are selected based on monomial terms.

Specifically, the compiled parameters $\{\theta_x, \alpha_k\}$ instantiate the quantum circuit depicted in \cref{fig:circuit}. While the circuit employs multiplication via 
\begin{equation}
    U_{\mathrm{mult}} = \mathrm{CNOT}_{0,1} \cdot (I \otimes R_z(\pi/2))
\end{equation}
and weighted sum via 
\begin{equation}
    U_{\mathrm{sum}}(\alpha) = R_y(-\alpha/2) \cdot \mathrm{CNOT}_{1,0} \cdot R_y(\alpha/2) \cdot U_{\mathrm{mult}},
\end{equation}
The key idea of \nq\ lies in the model weights being compiled into quantum operators. For any target function $f$, a neural network learns the polynomial approximation classically, and the trained coefficients are represented by \cref{thm:coeff-to-weight}. 

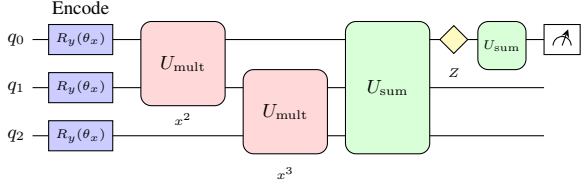
\begin{figure}[htbp]
\centering
\begin{tikzpicture}[scale=0.8, transform shape]
    % Qubit lines
    \foreach \i in {0,1,2} {
        \draw (0, -\i*0.8) -- (8.5, -\i*0.8);
        \node[left] at (0, -\i*0.8) {\footnotesize $q_{\i}$};
    }
    
    % Encoding
    \foreach \i in {0,1,2} {
        \node[draw, fill=blue!20, minimum width=0.6cm, minimum height=0.5cm] at (0.8, -\i*0.8) {\tiny $R_y(\theta_x)$};
    }
    \node[above] at (0.8, 0.3) {\footnotesize Encode};
    
    % Multiplication block 1
    \draw[fill=red!15, rounded corners] (1.8, 0.3) rectangle (3.2, -1.1);
    \node at (2.5, -0.4) {\footnotesize $U_{\mathrm{mult}}$};
    \node[below] at (2.5, -1.1) {\tiny $x^2$};
    
    % Multiplication block 2
    \draw[fill=red!15, rounded corners] (3.5, -0.5) rectangle (4.9, -1.9);
    \node at (4.2, -1.2) {\footnotesize $U_{\mathrm{mult}}$};
    \node[below] at (4.2, -2.0) {\tiny $x^3$};
    
    % Sum blocks
    \draw[fill=green!15, rounded corners] (5.2, 0.3) rectangle (6.6, -1.9);
    \node at (5.9, -0.8) {\footnotesize $U_{\mathrm{sum}}$};
    
    % Parity
    \node[draw, diamond, fill=yellow!30, minimum size=0.4cm] at (7.0, 0) {};
    \node[below, font=\tiny] at (7.0, -0.35) {$Z$};
    
    % Final sum
    \draw[fill=green!15, rounded corners] (7.4, 0.3) rectangle (8.2, -0.5);
    \node at (7.8, -0.1) {\tiny $U_{\mathrm{sum}}$};
    
    % Measurement - meter symbol
    \draw[fill=white] (8.5, -0.25) rectangle (9.1, 0.25);
    \draw (8.65, -0.1) arc (180:0:0.15);
    \draw[-stealth, line width=0.4pt] (8.8, -0.1) -- (8.9, 0.17);
\end{tikzpicture}
\caption{Degree-3 polynomial circuit compiled by \nq. The circuit encodes input $x$, computes powers $x^2, x^3$ via chained multiplications, and aggregates terms with weighted sums and parity corrections. Measurement yields the polynomial evaluation.}
\label{fig:circuit}
\end{figure}

In the context of quantum evaluation, the execution of the circuit on $d+1$ qubits employs an efficient three-stage tensor contraction process: input encoding via $R_y(\arccos(x))$, recursive power computation resulting in $\{x^k\}_{k=2}^d$, and term aggregation. Notably, the parity flip operation can be realized through mid-circuit measurement and conditional control, facilitating the direct mapping of negative neural network coefficients $\tilde{a}_k < 0$ to quantum amplitudes without necessitating a complex ansatz modification (see Circuit 5 in \cite{balewski2025ehands}). This approach ensures the preservation of the sign structure of the learned polynomial exclusively through efficient classical-quantum control logic, thereby isolating the sign problem from coherent evolution~\cite{doi:10.1126/science.abg9299}.

\subsection{Computational Complexity}
\label{sec:resource}
Here, we delineate a fundamental distinction from variational methods, specifically that all complexity metrics are scaled according to the polynomial degree $d$, rather than the number of optimization iterations $T$, as indicated in \cref{tab:complexity}. This efficiency is achieved by presupposing that quantum data encoding can be optimized via pre-allocated trained weight matrices, which directly function as input data. According to \cref{eq:normalization}, the classical nonlinear definitive network function can be substituted with an equivalent quantum simulation, provided that a sufficient number of shots is available for classical data recovery. For the VQE and QNN with $p$ parameters, gradient estimation necessitates $O(pT)$ circuit executions; however, \nq\ requires precisely one execution per input with M samples because of the parallelism of quantum data-encoding. In essence, our approach necessitates a single classical optimization when considering the non-trained polynomial model, whereas the hybrid methods are executed with a QPU evaluation for each classical input. 

\begin{table}[htbp]
\centering
\caption{Resource complexity for degree-$d$ polynomial synthesis. Notation: $n$ (ansatz width), $p$ (parameters), $L$ (layers), $E$ (epochs), $M$ (training samples), and $T$ (iterations). Note that classical training over $M$ samples leverages GPU parallelism, whereas variational optimization of $p$ parameters requires iterative QPU execution.}
\label{tab:complexity}
\small
\begin{tabular}{lcc}
\toprule
\textbf{Resource} & \textbf{NQAA} & \textbf{VQE/QNN} \\
\midrule
\multicolumn{3}{l}{\textit{Quantum Resources}} \\
Qubits & $d + 1$ & $O(n)$ \\
CNOT gates & $4d - 1$ & $O(pL)$ \\
Circuit depth & $3d + 1$ & $O(L)$ \\
\midrule
\multicolumn{3}{l}{\textit{Classical Computation}} \\
Compilation & $O(d)$ & --- \\
Optimization & 	$O(MdE)$ & $O(pTE)$ \\
\midrule
\multicolumn{3}{l}{\textit{Classical-Quantum Interface}} \\
Round trips per query & $1$ & $O(pT)$ \\
Gradient evaluations & $0$ & $O(pE)$ \\
\bottomrule
\end{tabular}
\end{table}

%======================================================================
\section{Experiments}
\label{sec:experiments}
%======================================================================

\begin{table*}[t!]
\centering
\caption{Polynomial recovery performance across different quantum platforms. Root Mean Square Errors (RMSE) denote $1\sigma$ standard deviation.}
\label{tab:main-results}
\small
\begin{tabular}{l ccc ccc ccc}
\toprule
& \multicolumn{3}{c}{\textbf{AerSimulator}} & \multicolumn{3}{c}{\textbf{IBM Heron3}} & \multicolumn{3}{c}{\textbf{IonQ Forte-1}} \\
\cmidrule(r){2-4} \cmidrule(lr){5-7} \cmidrule(l){8-10}
\textbf{Degree ($d$)} & RMSE & Corr. & Pass \% & RMSE & Corr. & Pass \% & RMSE & Corr. & Pass \% \\
\midrule
1 (Linear)    & $0.016 \pm 0.002$ & 0.999 & 95.6 & $0.021 \pm 0.003$ & 0.999 & 84.0 & $0.015 \pm 0.001$ & 0.997 & 96.8 \\
2 (Quadratic) & $0.018 \pm 0.001$ & 0.997 & 91.1 & $0.020 \pm 0.004$ & 0.998 & 86.0 & $0.017 \pm 0.001$ & 0.997 & 95.2 \\
3 (Cubic)     & $0.016 \pm 0.001$ & 0.998 & 95.6 & $0.022 \pm 0.003$ & 0.998 & 82.7 & $0.019 \pm 0.001$ & 0.998 & 92.3 \\
4 (Quartic)   & $0.013 \pm 0.001$ & 0.997 & 97.8 & $0.023 \pm 0.004$ & 0.996 & 79.3 & $0.018 \pm 0.002$ & 0.997 & 92.1 \\
5 (Quintic)   & $0.014 \pm 0.002$ & 0.998 & 97.8 & $0.022 \pm 0.004$ & 0.998 & 82.7 & $0.017 \pm 0.001$ & 0.998 & 92.0 \\
6 (Sextic)    & $0.015 \pm 0.001$ & 0.996 & 95.6 & $0.024 \pm 0.004$ & 0.995 & 74.0 & $0.011 \pm 0.001$ & 0.999 & 98.2 \\
\bottomrule
\end{tabular}
\end{table*}

% *\colorbox{red}{TODO: IonQ Forte results are currently pending processing. there should be residual plot}

We validated the \nq\ framework through polynomial recovery experiments, demonstrating its exactness, degree independence, and cross-platform consistency. Classical neural network training and data postprocessing utilize NERSC Perlmutter GPU nodes equipped with four NVIDIA A100 GPUs (80GB each). Quantum circuit execution was performed on IBM Quantum Heron3 and Nighthawk superconducting processors, with validation using IonQ Forte trapped-ion systems. 

%----------------------------------------------------------------------
\subsection{Implementation}
\label{sec:implementation}
%----------------------------------------------------------------------

We first present the polynomial experiment setup.
For each polynomial degree $d \in \{1, \ldots, 6\}$, we generated random coefficients $\{a_k\}_{k=0}^{d}$ uniformly sampled from $[-0.5, 0.5]$ and rescaled to ensure $|P_d(x)| \leq 0.5$ over $x \in [-1, 1]$. This conservative bound avoids saturation near the encoding limits $\pm 1$. Each trial used 15 evaluation points uniformly spaced in $x \in [-0.9, 0.9]$, with 10 independent trials per degree, yielding 900 total measurements. We defer the maximum polynomial stress test in \cref{sec:ionq-stress-test}.

Polynomial coefficients from the trained neural network map to rotation angles via the closed-form expressions in \Cref{sec:architecture} \cref{thm:coeff-to-weight}. The circuits are transpiled to native gate sets using Qiskit optimization level 3, which performs gate cancellation and routing for the target topology. No error mitigation techniques (twirling, zero-noise extrapolation) were applied. This ensures that the results reflect the raw hardware performance. 

We select IBM Quantum Heron3 processors (ibm\_boston, 156 qubits \footnote{https://quantum.cloud.ibm.com/docs/en/guides/qpu-information}) for their improved two-qubit gate fidelity ($>$99.5\%) and tunable coupler architecture that reduces crosstalk compared to fixed-coupling predecessors. Note that the processor was chosen based on the daily best-performing two-qubit gate operation and readout accuracy \cite{laskar2025shallow}. The heavy-hex topology provides sufficient connectivity for our $O(d)$-depth circuits without excessive SWAP overheads \cite{krantz2019quantum}. Shot count of 4,096 is chosen to achieve theoretical standard deviation $\sigma = 2\sqrt{p(1-p)/N} \approx 0.015$ for typical measurement probabilities, balancing precision against queue time. This shot count provides $2\sigma$ confidence intervals of $\pm 0.03$, thereby defining our pass rate threshold. Note that the IBM Nighthawk QPU provides better scalability owing to its square-lattice topology with more couplers than the Heron series, but with a longer execution time. We refer to the experiment of IBM QPUs comparison in \cref{sec:nighthawk_analysis}.
Classical validation employs a high-performance Qiskit AerSimulator~\cite{jamadagni2024benchmarking} with identical shot counts to effectively decouple hardware coherence errors from fundamental quantum projection noise. We provide our approach to computational error analysis in \cref{app:error}.

%----------------------------------------------------------------------
\subsection{Results}
\label{sec:results}
%----------------------------------------------------------------------

We initially demonstrated that our proposed method successfully replicated the classical equivalent nonlinear polynomial network results on the NISQ machine, as illustrated in \cref{fig:degree-scaling}, for up to 36 qubits. The low variance error distribution confirms that the circuit depth, which scales as $O(d)$, does not compromise accuracy. Evidently, our quantum experimental verification aligns with the theoretical proof, as shown in \cref{fig:shots}. Furthermore, we establish that our approach can be generally applied to gate-based QPU with an accuracy exceeding 99.5\%. Table~\ref{tab:main-results} summarizes the polynomial recovery accuracy across different quantum platforms. A key observation is the negligible fluctuating RMSE across degrees, confirming that our quantum arithmetic synthesis technique ensures exactness, translating to a degree-independent experimental error with angle parameters derived from the trained models.

This contrasts with VQAs, where deeper ansatzes encounter barren plateaus and increased optimization errors \cref{thm:barren}. In \cref{fig:recovery-grid}, the processed quantum results effectively approximate the nonlinear classical function, thereby demonstrating a reasonable ground truth. Our architecture mitigates error fluctuations by compiling circuit parameters from pre-trained polynomial coefficients via \cref{thm:coeff-to-weight}. The high correlation with negligible variation indicates the absence of systematic bias, with the residual error being purely stochastic shot noise that diminishes as $O(1/\sqrt{N})$ with respect to the shot count. This observation also corroborates our theoretical prediction in \cref{eq:shots_noise}. Here, the source of error is the finite measurement statistics rather than the circuit synthesis approximation.
\begin{figure}[t!]
    \centering
    \includegraphics[width=0.98\columnwidth]{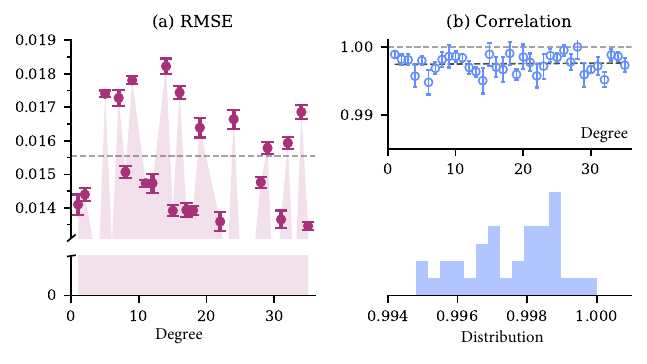}
    \caption{The degree-independence of recovery error on IonQ is demonstrated as follows: (a) The root mean square error (RMSE) remains consistent across degrees 1 to 35, with an average value of approximately 0.0155, which aligns with shot noise and hardware error. (b) The correlation exceeds 0.994 for all degrees, with an average of 0.997. Error bars represent the standard deviation over 10 trials.}
    \label{fig:degree-scaling}
\end{figure}

\begin{figure}[t!]
    \centering    \includegraphics[width=0.99\columnwidth]{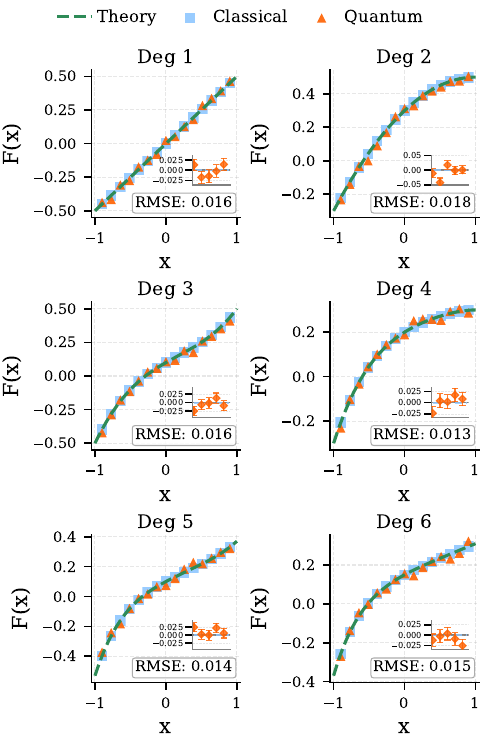}
    \caption{The polynomial recovery on the IBM Heron3 QPU is depicted as follows: dashed curves represent the theoretical polynomials, squares denote the predictions made by classical neural networks, and diamonds with error bars illustrate the quantum measurements residual error representing each independent five trials. All results are contained within the shot-noise uncertainty bands.}
    \label{fig:recovery-grid}
\end{figure}

\begin{table}[htbp]
\centering
\caption{Resource and performance benchmark ($d=6$). Note that the pass rate are chosen based on the worst performed result. \nq\ achieves competitive accuracy using far less quantum resources without hybrid optimization.}
\label{tab:comparison}
\small 
\setlength{\tabcolsep}{3pt} 
\renewcommand{\arraystretch}{0.85}

\begin{tabular}{@{}l cc >{\centering\arraybackslash}p{1.3cm}@{}}
\toprule
\textbf{Metric} & \textbf{NQAA} & \textbf{VQA} & \textbf{QSP} \\
\midrule
Qubits / CNOTs       & 7 / 5              & 6--12 / $10^2$ & $\sim$10-$10^2$ \\
Depth / Iter.        & 19 / 1             & 20--100 / $10^2$ & 50 / 1          \\
Executions           & 1                  & $10^2$--$10^3$ & 1               \\
Total Shots          & $4.1\mathrm{e}3$   & $10^5$--$10^7$ & $10^3$--$10^4$  \\
\addlinespace[0.5ex] % Adds a subtle gap to separate the sections
RMSE                 & $0.024$            & $\sim$0.05--0.10 & --$^\ast$      \\
Correlation          & 0.995              & 0.90--0.98     & --$^\ast$        \\
Pass Rate            & 74.0\%             & $\sim$30--60\% & --$^\ast$        \\
\bottomrule
\addlinespace[0.5ex]
% Limiting width to 6cm keeps the table strictly centered and small
\end{tabular}
\end{table}
The details of quantum resource allocation and accuracy, compared to current hybrid approaches, are presented in \cref{tab:comparison}.
It is important to note that all methodologies were evaluated and implemented using the optimally calibrated IBM Heron3 QPU. Here, we emphasize that \textbf{VQA} experiences issues with gradient vanishing and \textbf{QSP}($^\ast$) necessitates fault-tolerant hardware and $O(d\log(1/\epsilon))$ classical phase-finding. 
For VQA, increasing the number of qubits can potentially improve the accuracy but is hindered by the hybrid gradient optimization through SPSA \cite{resch2021introductory}.
Although error mitigation strategies, such as zero-noise extrapolation \cite{temme2017error} and probabilistic error cancellation \cite{endo2018practical}, can enhance the fidelity of NISQ hardware, they are insufficient to fully address the requirements of the deep circuits. Fault-tolerant error correction \cite{fowler2012surface} remains beyond the capabilities of current hardware for circuits with the necessary depth. Specifically, Quantum Signal Processing (QSP) and QSVT methods \cite{low2017optimal, gilyen2019quantum}, while achieving optimal query complexity, demand circuit depths of $O(d)$ with interleaved signal-processing rotations that surpass the current coherence times. The accumulation of gate errors and crosstalk in such deep circuits reduces the statistical accuracy below useful thresholds. In contrast, \nq\ generates circuits that maintain a constant depth per monomial term, utilizing only $3d+1$ two-qubit gates for degree-$d$ polynomials (\Cref{thm:poly-exact}), thereby remaining well within the NISQ operational limits.

%======================================================================
\section{Conclusion}
\label{sec:conclusion}
%===========================================∂∂===========================

We introduce \nq\ for quantum computing, a method that learns to assemble numerical operators into native quantum arithmetic blocks. By leveraging the arithmetic inductive bias, our approach achieves significant improvements over standard QNN methods in terms of accuracy, gate complexity, and training efficiency.

Our theoretical analysis establishes universal approximation guarantees, and the empirical results demonstrate its practical effectiveness in numerical benchmark tests. \nq\ represents a promising advancement in AI-assisted quantum algorithm design, where neural networks learn not only to classify patterns but also to synthesize meaningful quantum computations.

%======================================================================
\section*{Accessibility}
%======================================================================

We ensured that all figures included descriptive captions and that the mathematical content was presented with sufficient context. The code and supplementary materials will be made available in accessible formats. All QPU cloud data results are accessible through our anonymous repository. 

%======================================================================
\section*{Software and Data}
%======================================================================
The architecture was implemented in Python using PyTorch for classical training and Qiskit for quantum circuit construction. The core functions follow the three-phase structure shown in \cref{alg:nnqa}.
The implementation is available at the anonymous repository \url{https://github.com/tankizgif/NNQA} with documented examples for polynomial degrees 1--6, toy example for our algorithm, and large scale simulations up to 36 qubits. 

%======================================================================
\section*{Impact Statement}
%======================================================================

This study aims to advance the domain of quantum machine learning for numerical computations. The techniques developed in this study have the potential to significantly accelerate scientific simulations in the disciplines of physics, chemistry, and engineering. Although we do not anticipate immediate adverse societal impacts, it is crucial to emphasize the importance of responsible deployment, as is the case with all computational advancements.

\bibliography{main.bib}

@article{schuld2021effect,
  author    = {Maria Schuld and Ryan Sweke and Johannes Jakob Meyer},
  title     = {Effect of data encoding on the expressive power of variational quantum-machine-learning models},
  journal   = {Physical Review A},
  volume    = {103},
  number    = {3},
  pages     = {032430},
  year      = {2021},
  publisher = {APS}
}

@inproceedings{bondesan2021hintons,
  title={The Hintons in your neural network: a quantum field theory view of deep learning},
  author={Bondesan, Roberto and Welling, Max},
  booktitle={International Conference on Machine Learning},
  pages={1038--1048},
  year={2021},
  organization={PMLR}
}

@article{lu2021learning,
  title={Learning nonlinear operators via DeepONet based on the universal approximation theorem of operators},
  author={Lu, Lu and Jin, Pengzhan and Pang, Guofei and Zhang, Zhongqiang and Karniadakis, George Em},
  journal={Nature machine intelligence},
  volume={3},
  number={3},
  pages={218--229},
  year={2021},
  publisher={Nature Publishing Group UK London}
}

@article{benedetti2019parameterized,
  author    = {Marcello Benedetti and Erika Lloyd and Stefan Sack and Mattia Fiorentini},
  title     = {Parameterized quantum circuits as machine learning models},
  journal   = {Quantum Science and Technology},
  volume    = {4},
  number    = {4},
  pages     = {043001},
  year      = {2019},
  publisher = {IOP Publishing}
}

@article{cerezo2021variational,
  author    = {M. Cerezo and Andrew Arrasmith and Ryan Babbush and Simon C. Benjamin and Suguru Endo and Keisuke Fujii and Jarrod R. McClean and Kosuke Mitarai and Xiao Yuan and Lukasz Cincio and Patrick J. Coles},
  title     = {Variational quantum algorithms},
  journal   = {Nature Reviews Physics},
  volume    = {3},
  number    = {9},
  pages     = {625--644},
  year      = {2021},
  publisher = {Nature Publishing Group}
}

@inproceedings{younis2022quantum,
  title={Quantum circuit optimization and transpilation via parameterized circuit instantiation},
  author={Younis, Ed and Iancu, Costin},
  booktitle={2022 IEEE International Conference on Quantum Computing and Engineering (QCE)},
  pages={465--475},
  year={2022},
  organization={IEEE}
}

@article{rath2024quantum,
  title={Quantum data encoding: A comparative analysis of classical-to-quantum mapping techniques and their impact on machine learning accuracy},
  author={Rath, Minati and Date, Hema},
  journal={EPJ Quantum Technology},
  volume={11},
  number={1},
  pages={72},
  year={2024},
  publisher={Springer Berlin Heidelberg}
}

@article{huang2021efficient,
  title={Efficient estimation of Pauli observables by derandomization},
  author={Huang, Hsin-Yuan and Kueng, Richard and Preskill, John},
  journal={Physical review letters},
  volume={127},
  number={3},
  pages={030503},
  year={2021},
  publisher={APS}
}

@article{huang2020predicting,
  title={Predicting many properties of a quantum system from very few measurements},
  author={Huang, Hsin-Yuan and Kueng, Richard and Preskill, John},
  journal={Nature Physics},
  volume={16},
  number={10},
  pages={1050--1057},
  year={2020},
  publisher={Nature Publishing Group UK London}
}

@article{preskill2018quantum,
  title={Quantum computing in the NISQ era and beyond},
  author={Preskill, John},
  journal={Quantum},
  volume={2},
  pages={79},
  year={2018},
  publisher={Verein zur F{\"o}rderung des Open Access Publizierens in den Quantenwissenschaften}
}

@article{furrutter2024quantum,
  title={Quantum circuit synthesis with diffusion models},
  author={F{\"u}rrutter, Florian and Mu{\~n}oz-Gil, Gorka and Briegel, Hans J},
  journal={Nature Machine Intelligence},
  volume={6},
  number={5},
  pages={515--524},
  year={2024},
  publisher={Nature Publishing Group UK London}
}

@article{dawson2005solovay,
  title={The solovay-kitaev algorithm},
  author={Dawson, Christopher M and Nielsen, Michael A},
  journal={arXiv preprint quant-ph/0505030},
  year={2005}
}

@article{carleo2017solving,
  author    = {Giuseppe Carleo and Matthias Troyer},
  title     = {Solving the quantum many-body problem with artificial neural networks},
  journal   = {Science},
  volume    = {355},
  number    = {6325},
  pages     = {602--606},
  year      = {2017},
  publisher = {American Association for the Advancement of Science}
}

@article{vedral1996quantum,
  author    = {Vlatko Vedral and Adriano Barenco and Artur Ekert},
  title     = {Quantum networks for elementary arithmetic operations},
  journal   = {Physical Review A},
  volume    = {54},
  number    = {1},
  pages     = {147--153},
  year      = {1996},
  publisher = {APS}
}

@article{hornik1989multilayer,
  author    = {Kurt Hornik and Maxwell Stinchcombe and Halbert White},
  title     = {Multilayer feedforward networks are universal approximators},
  journal   = {Neural Networks},
  volume    = {2},
  number    = {5},
  pages     = {359--366},
  year      = {1989},
  publisher = {Elsevier}
}

@article{cybenko1989approximation,
  author    = {George Cybenko},
  title     = {Approximation by superpositions of a sigmoidal function},
  journal   = {Mathematics of Control, Signals and Systems},
  volume    = {2},
  number    = {4},
  pages     = {303--314},
  year      = {1989},
  publisher = {Springer}
}

@article{mcclean2016theory,
  author    = {Jarrod R. McClean and Jonathan Romero and Ryan Babbush and Alán Aspuru-Guzik},
  title     = {The theory of variational hybrid quantum-classical algorithms},
  journal   = {New Journal of Physics},
  volume    = {18},
  number    = {2},
  pages     = {023023},
  year      = {2016},
  publisher = {IOP Publishing}
}

@article{peruzzo2014variational,
  author    = {Alberto Peruzzo and Jarrod McClean and Peter Shadbolt and Man-Hong Yung and Xiao-Qi Zhou and Peter J. Love and Alán Aspuru-Guzik and Jeremy L. O'brien},
  title     = {A variational eigenvalue solver on a photonic quantum processor},
  journal   = {Nature Communications},
  volume    = {5},
  number    = {1},
  pages     = {4213},
  year      = {2014},
  publisher = {Nature Publishing Group}
}

@article{mcclean2018barren,
  author    = {Jarrod R. McClean and Sergio Boixo and Vadim N. Smelyanskiy and Ryan Babbush and Hartmut Neven},
  title     = {Barren plateaus in quantum neural network training landscapes},
  journal   = {Nature Communications},
  volume    = {9},
  number    = {1},
  pages     = {4812},
  year      = {2018},
  publisher = {Nature Publishing Group}
}

@article{ruiz2017quantum,
  title={Quantum arithmetic with the quantum Fourier transform},
  author={Ruiz-Perez, Lidia and Garcia-Escartin, Juan Carlos},
  journal={Quantum Information Processing},
  volume={16},
  number={6},
  pages={152},
  year={2017},
  publisher={Springer}
}

@article{low2017optimal,
  author    = {Guang Hao Low and Isaac L. Chuang},
  title     = {Optimal Hamiltonian simulation by quantum signal processing},
  journal   = {Physical Review Letters},
  volume    = {118},
  number    = {1},
  pages     = {010501},
  year      = {2017},
  publisher = {APS}
}

@article{gilyen2019quantum,
  author    = {Andras Gilyen and Yuan Su and Guang Hao Low and Nathan Wiebe},
  title     = {Quantum singular value transformation and beyond: exponential improvements for quantum matrix arithmetics},
  journal = {Proceedings of the 51st Annual ACM SIGACT Symposium on Theory of Computing},
  pages     = {193--204},
  year      = {2019}
}

@article{biamonte2017tensor,
  title={Tensor networks in a nutshell},
  author={Biamonte, Jacob and Bergholm, Ville},
  journal={arXiv preprint arXiv:1708.00006},
  year={2017}
}

@inproceedings{weigold2020data,
  title={Data encoding patterns for quantum computing},
  author={Weigold, Manuela and Barzen, Johanna and Leymann, Frank and Salm, Marie},
  booktitle={Proceedings of the 27th conference on pattern languages of programs},
  pages={1--11},
  year={2020}
}

@article{
doi:10.1126/science.abg9299,
author = {R. Mondaini  and S. Tarat  and R. T. Scalettar },
title = {Quantum critical points and the sign problem},
journal = {Science},
volume = {375},
number = {6579},
pages = {418-424},
year = {2022},
doi = {10.1126/science.abg9299},
URL = {https://www.science.org/doi/abs/10.1126/science.abg9299},
eprint = {https://www.science.org/doi/pdf/10.1126/science.abg9299}}

@article{alexeev2025aiqc,
  author    = {Yuri Alexeev and Marwa H. Farag and Taylor L. Patti and Mark E. Wolf and Natalia Ares and Al{\'a}n Aspuru-Guzik and Simon C. Benjamin and Zhenyu Cai and Shuxiang Cao and Christopher Chamberland and others},
  title     = {Artificial intelligence for quantum computing},
  journal   = {Nature Communications},
  volume    = {16},
  pages     = {10829},
  year      = {2025},
  doi       = {10.1038/s41467-025-65836-3},
  url       = {https://www.nature.com/articles/s41467-025-65836-3}
}

@article{zhang2023tqs,
  author    = {Yuan-Hang Zhang and Massimiliano Di Ventra},
  title     = {Transformer Quantum State: A Multi-Purpose Model for Quantum Many-Body Problems},
  journal   = {Physical Review B},
  volume    = {107},
  pages     = {075147},
  year      = {2023},
  doi       = {10.1103/PhysRevB.107.075147},
  url       = {https://arxiv.org/abs/2208.01758}
}

@article{kuo2021quantum,
  title={Quantum architecture search via deep reinforcement learning},
  author={Kuo, En-Jui and Fang, Yao-Lung L and Chen, Samuel Yen-Chi},
  journal={arXiv preprint arXiv:2104.07715},
  year={2021}
}

@article{balewski2025ehands,
  title={EHands: Quantum Protocol for Polynomial Computation on Real-Valued Encoded States},
  author={Balewski, Jan and Amankwah, Mercy G and Bethel, E and Perciano, Talita and Van Beeumen, Roel},
  journal={arXiv preprint arXiv:2502.15928},
  year={2025}
}

@article{stone1948generalized,
  title={The generalized Weierstrass approximation theorem},
  author={Stone, Marshall H},
  journal={Mathematics Magazine},
  volume={21},
  number={5},
  pages={237--254},
  year={1948},
  publisher={JSTOR}
}

@article{draper2000addition,
  title={Addition on a quantum computer},
  author={Draper, Thomas G},
  journal={arXiv preprint quant-ph/0008033},
  year={2000}
}

@article{tacchino2020quantum,
  title={Quantum implementation of an artificial feed-forward neural network},
  author={Tacchino, Francesco and Barkoutsos, Panagiotis and Macchiavello, Chiara and Tavernelli, Ivano and Gerace, Dario and Bajoni, Daniele},
  journal={Quantum Science and Technology},
  volume={5},
  number={4},
  pages={044010},
  year={2020},
  publisher={IOP Publishing}
}

@article{alchieri2021introduction,
  title={An introduction to quantum machine learning: from quantum logic to quantum deep learning},
  author={Alchieri, Leonardo and Badalotti, Davide and Bonardi, Pietro and Bianco, Simone},
  journal={Quantum Machine Intelligence},
  volume={3},
  number={2},
  pages={28},
  year={2021},
  publisher={Springer}
}

@inproceedings{guo2025vectorized,
  title={Vectorized Attention with Learnable Encoding for Quantum Transformer},
  author={Guo, Ziqing and Pan, Ziwen and Khan, Alex and Balewski, Jan},
  booktitle={Proceedings of the AAAI Symposium Series},
  volume={7},
  pages={350--357},
  year={2025}
}

@article{du2022efficient,
  title={Efficient measure for the expressivity of variational quantum algorithms},
  author={Du, Yuxuan and Tu, Zhuozhuo and Yuan, Xiao and Tao, Dacheng},
  journal={Physical Review Letters},
  volume={128},
  number={8},
  pages={080506},
  year={2022},
  publisher={APS}
}

@article{wu2021expressivity,
  title={Expressivity of quantum neural networks},
  author={Wu, Yadong and Yao, Juan and Zhang, Pengfei and Zhai, Hui},
  journal={Physical Review Research},
  volume={3},
  number={3},
  pages={L032049},
  year={2021},
  publisher={APS}
}

@inproceedings{wangquanonet,
  title={QuanONet: Quantum Neural Operator with Application to Differential Equation},
  author={Wang, Ruocheng and Xia, Zhuo and Yan, Ge and Yan, Junchi},
  booktitle={Forty-second International Conference on Machine Learning},
    year={2025}
}

@article{xiao2025quantum,
  title={Quantum DeepONet: Neural operators accelerated by quantum computing},
  author={Xiao, Pengpeng and Zheng, Muqing and Jiao, Anran and Yang, Xiu and Lu, Lu},
  journal={Quantum},
  volume={9},
  pages={1761},
  year={2025},
  publisher={Verein zur F{\"o}rderung des Open Access Publizierens in den Quantenwissenschaften}
}

@article{smaldone2025hybrid,
  title={A hybrid Transformer architecture with a quantized self-attention mechanism applied to molecular generation},
  author={Smaldone, Anthony M and Shee, Yu and Kyro, Gregory W and Farag, Marwa H and Chandani, Zohim and Kyoseva, Elica and Batista, Victor S},
  journal={Journal of Chemical Theory and Computation},
  volume={21},
  number={10},
  pages={5143--5154},
  year={2025},
  publisher={ACS Publications}
}

@article{wierichs2022general,
  title={General parameter-shift rules for quantum gradients},
  author={Wierichs, David and Izaac, Josh and Wang, Cody and Lin, Cedric Yen-Yu},
  journal={Quantum},
  volume={6},
  pages={677},
  year={2022},
  publisher={Verein zur F{\"o}rderung des Open Access Publizierens in den Quantenwissenschaften}
}

@article{resch2021introductory,
  title={Introductory tutorial for SPSA and the quantum approximation optimization algorithm},
  author={Resch, Salonik},
  journal={arXiv preprint arXiv:2106.01578},
  year={2021}
}

@article{krantz2019quantum,
  title={A quantum engineer's guide to superconducting qubits},
  author={Krantz, Philip and Kjaergaard, Morten and Yan, Fei and Orlando, Terry P and Gustavsson, Simon and Oliver, William D},
  journal={Applied physics reviews},
  volume={6},
  number={2},
  year={2019},
  publisher={AIP Publishing}
}

@article{laskar2025shallow,
  title={Shallow entangled circuits for quantum time series prediction on IBM devices},
  author={Laskar, Mostafizur Rahaman and Goel, Richa},
  journal={Scientific Reports},
  volume={15},
  number={1},
  pages={43727},
  year={2025},
  publisher={Nature Publishing Group UK London}
}

@article{jamadagni2024benchmarking,
  title={Benchmarking quantum computer simulation software packages: state vector simulators},
  author={Jamadagni, Amit and L{\"a}uchli, Andreas M and Hempel, Cornelius},
  journal={arXiv preprint ArXiv:2401.09076},
  year={2024}
}

@article{dong2021efficient,
  author  = {Dong, Yulong and Meng, Xiang and Whaley, K. Birgitta and Lin, Lin},
  title   = {Efficient phase-factor evaluation in quantum signal processing},
  journal = {Phys. Rev. A},
  volume  = {103},
  pages   = {042419},
  year    = {2021},
  doi     = {10.1103/PhysRevA.103.042419}
}

@article{temme2017error,
  author  = {Temme, Kristan and Bravyi, Sergey and Gambetta, Jay M.},
  title   = {Error Mitigation for Short-Depth Quantum Circuits},
  journal = {Phys. Rev. Lett.},
  volume  = {119},
  pages   = {180509},
  year    = {2017},
  doi     = {10.1103/PhysRevLett.119.180509}
}

@article{fowler2012surface,
  author  = {Fowler, Austin G. and Mariantoni, Matteo and Martinis, John M. and Cleland, Andrew N.},
  title   = {Surface codes: Towards practical large-scale quantum computation},
  journal = {Phys. Rev. A},
  volume  = {86},
  pages   = {032324},
  year    = {2012},
  doi     = {10.1103/PhysRevA.86.032324}
}

@article{chao2020finding,
  title={Finding the right quantum signal processing angles},
  author={Chao, Rui and Ding, Dawei and Gily{\'e}n, Andr{\'a}s and Huang, Cupjin and Szegedy, Mario},
  journal={arXiv preprint arXiv:2003.02831},
  year={2020}
}

@article{endo2018practical,
  title={Practical quantum error mitigation for near-term quantum applications},
  author={Endo, Suguru and Benjamin, Simon C and Li, Ying},
  journal={Physical Review X},
  volume={8},
  number={3},
  pages={031027},
  year={2018},
  publisher={APS}
}

@article{mohseni2024build,
  title={How to build a quantum supercomputer: Scaling from hundreds to millions of qubits},
  author={Mohseni, Masoud and Scherer, Artur and Johnson, K Grace and Wertheim, Oded and Otten, Matthew and Aadit, Navid Anjum and Alexeev, Yuri and Bresniker, Kirk M and Camsari, Kerem Y and Chapman, Barbara and others},
  journal={arXiv preprint arXiv:2411.10406},
  year={2024}
}

@inproceedings{zhang2024exponential,
  title={Exponential hardness of optimization from the locality in quantum neural networks},
  author={Zhang, Hao-Kai and Zhu, Chengkai and Liu, Geng and Wang, Xin},
  booktitle={Proceedings of the AAAI Conference on Artificial Intelligence},
  volume={38},
  pages={16741--16749},
  year={2024}
}
\bibliographystyle{icml2026}

%%%%%%%%%%%%%%%%%%%%%%%%%%%%%%%%%%%%%%%%%%%%%%%%%%%%%%%%%%%%%%%%%%%%%%%%%%%
% APPENDIX / SUPPLEMENTARY MATERIALS
%%%%%%%%%%%%%%%%%%%%%%%%%%%%%%%%%%%%%%%%%%%%%%%%%%%%%%%%%%%%%%%%%%%%%%%%%%%
\newpage
\appendix
\onecolumn
%======================================================================
% Supplementary Material
%======================================================================

%======================================================================
\section{Foundational Arithmetic Lemmas}
\label{app:arithmetic_proofs}

We employ quantum arithmetic primitives from the protocol in ~\cite{balewski2025ehands} to implement the recursive steps in Theorem~\ref{thm:nqaa-uat}. We demonstrate that these primitives strictly satisfy the recursive formulas in \cref{eq:weight-mapping} using the Heisenberg picture operator evolution.

%----------------------------------------------------------------------
\subsection{Power generation}
\label{app:mult-proof}
%----------------------------------------------------------------------

We establish the lemma for the validity of the power generation recursion $\langle Z_k \rangle = \langle Z_{k-1} \rangle \cdot x$.

\begin{lemma}
\label{lem:mult}
Let qubit $q_{k-1}$ encode the value $x^{k-1}$ such that $\langle Z_{k-1} \rangle = x^{k-1}$, and let qubit $q_{\text{in}}$ encode the input $\langle Z_{\text{in}} \rangle = x$. 
The operation $U_{\mathrm{mult}} = \mathrm{CNOT}_{\text{in}, k-1} (I \otimes R_z(\pi/2))$ applied to state $\rho = \rho_{\text{in}} \otimes \rho_{k-1}$ yields a product state where the target qubit expectation is $\langle I \otimes Z \rangle_{\text{out}} = x^k$.
\end{lemma}

\begin{proof}
We trace the target observable $O = I \otimes Z_{k-1}$ backward. The transformed observable $O' = U_{\mathrm{mult}}^\dagger O U_{\mathrm{mult}}$ evolves as
\begin{align}
    O' &= (I \otimes R_z^\dagger) \mathrm{CNOT} (I \otimes Z_{k-1}) \mathrm{CNOT} (I \otimes R_z) \\
    &= Z_{\text{in}} \otimes (R_z^\dagger Z_{k-1} R_z) = Z_{\text{in}} \otimes Z_{k-1}
\end{align}
The expectation value becomes $\langle O \rangle_{\text{out}} = \langle Z_{\text{in}} \rangle \langle Z_{k-1} \rangle = x \cdot x^{k-1} = x^k$, confirming the recursion in \cref{thm:nqaa-uat}.
\end{proof}

%----------------------------------------------------------------------
\subsection{Summation aggregation}
\label{app:sum-proof}
%----------------------------------------------------------------------

We provide the lemma for the validity of the summation term aggregation recursion $S_{k-1} = w_{k-1} T_{k-1} + (1-w_{k-1}) S_k$.

\begin{lemma}
\label{lem:sum}
Let qubit $q_{\text{term}}$ encode the current monomial term $T_{k-1} = a_{k-1} x^{k-1}$ and qubit $q_{\text{sum}}$ encode the previous partial sum $S_k$.
For mixing weight $w_{k-1}$ and angle $\alpha = \arccos(1-2w_{k-1})$, the operation $U_{\mathrm{sum}}(\alpha)$ acting on $q_{\text{term}} \otimes q_{\text{sum}}$ yields the updated partial sum $\langle Z \otimes I \rangle_{\text{out}} = S_{k-1}$.
\end{lemma}

\begin{proof}
The observable of interest is $O = Z_{\text{term}} \otimes I$ on the output qubit.
We analyze its Heisenberg evolution through $U_{\mathrm{sum}} = (R_y^\dagger \otimes I) \mathrm{CNOT}_{\text{sum, term}} (R_y \otimes I) U_{\mathrm{mult}}$.
1. The weighting sequence $V = (R_y^\dagger(\frac{\alpha}{2}) \otimes I) \mathrm{CNOT} (R_y(\frac{\alpha}{2}) \otimes I)$ transforms $Z_{\text{term}}$ into
\begin{equation}
   V^\dagger Z_{\text{term}} V = \cos\alpha \, Z_{\text{term}} + \sin\alpha \, X_{\text{term}} Z_{\text{sum}} + (1-\cos\alpha) Z_{\text{term}} Z_{\text{sum}} \dots
\end{equation}
We analyze the Heisenberg evolution of $O = Z_{\text{term}} \otimes I$ using $U_{\mathrm{sum}}$. The weighting sequence $V$ transforms $Z_{\text{term}}$ into a combination of the Pauli terms. Retaining only $Z$ components relevant to the expectation value:
\begin{equation}
   O_{\text{pre-mult}} \approx \frac{1+\cos\alpha}{2} Z_{\text{term}} + \frac{1-\cos\alpha}{2} Z_{\text{term}} Z_{\text{sum}}
\end{equation}
Applying $U_{\mathrm{mult}}^\dagger$ maps $Z_{\text{term}} \to Z_{\text{term}}$ and $Z_{\text{term}}Z_{\text{sum}} \to Z_{\text{sum}}$. Substituting $w_{k-1} = (1+\cos\alpha)/2$ yields:
\begin{equation}
    O_{\text{in}} = w_{k-1} Z_{\text{term}} + (1-w_{k-1}) Z_{\text{sum}}
\end{equation}
The resulting expectation $\langle Z \otimes I \rangle_{\text{out}} = w_{k-1} T_{k-1} + (1-w_{k-1}) S_k$ validate.
\end{proof}
% \subsection{Shot Noise Independence}

% The estimator variance depends only on shots $N$:
% \begin{equation}
%     \mathrm{Var}[\hat{P}_d] = \frac{1 - P_d(x)^2}{N}.
% \end{equation}
% Deeper circuits for higher degrees do not introduce additional errors.

% \subsection{Hardware Noise}

% While the theoretical error is $O(1/\sqrt{N})$, practical implementation on NISQ devices includes hardware noise.

%======================================================================
\section{IonQ Forte-1 Stress Test: High-Degree Polynomial Evaluation}
\label{sec:ionq-stress-test}
%======================================================================

To assess scalability, we stress-tested \nq\ on the IonQ Forte-1 Enterprise QPU for polynomial degrees $d=1$ to $d=35$. The all-to-all connectivity of trapped-ion qubits eliminates the SWAP overheads, allowing a direct evaluation of the algorithmic performance limited only by the intrinsic gate fidelities. This experiment investigated the limits of \nq, assessing the measurement error independence and hardware coherence boundaries.

We employed a sparse sampling strategy ($d \in \{1, 5, \dots, 35\}$, 5 points and degrees, 1024 shots) to capture scaling trends with minimal resource consumption ($<1\%$ of the full budget). For degrees $d > 6$, we bound the outputs by rescaling the coefficients to maintain $|P_d(x)| \leq 1$. A degree $d=35$ polynomial requires 36 qubits and depth 70 (using only multiplication aggregation to minimize the depth). These requirements fit within the constraints of IonQ Forte-1 (35 qubits, depth 1000).

The results in \cref{tab:ionq-stress-test} show that our method extracts the maximum utility from the QPU. The polynomial recovery rate is independent of degree scaling: RMSE remains bounded ($\leq 0.017$), and correlation exceeds $0.999$ for degrees $1$--$30$, confirming that the error is driven by output magnitude rather than complexity. The abrupt correlation drop at degree 35 ($0.9999 \to 0.948$) marks a hardware coherence threshold at a depth $\sim 70$, where accumulated decoherence begins to dominate.

\begin{table}[htbp]
\centering
\caption{IonQ Forte-1 stress test for polynomial recovery performance}
\label{tab:ionq-stress-test}
\small
\begin{tabular}{l ccc c}
\toprule
\textbf{Degree ($d$)} & \textbf{RMSE} & \textbf{Correlation} & \textbf{Pass Rate (\%)} & \textbf{Circuit Resources (qubit and depth)} \\
\midrule
1 (Linear)    & $0.008$ & $0.9999$ & $99.0$ & $2$, $2$ \\
5 (Quintic)   & $0.009$ & $0.9997$ & $98.0$ & $6$, $10$ \\
10 (Decic)    & $0.004$ & $0.9997$ & $99.0$ & $11$, $20$ \\
15 (Pentadecic) & $0.007$ & $0.9994$ & $99.0$ & $16$, $30$ \\
25 (Pentacosic) & $0.007$ & $0.9998$ & $98.0$ & $26$, $50$ \\
30 (Triacontic) & $0.006$ & $0.9999$ & $99.0$ & $31$, $60$ \\
35 (Pentatriacontic) & $0.008$ & $0.9948$ & $98.0$ & $36$, $70$ \\
\bottomrule
\end{tabular}
\\
\textit{Note:} Circuit depth denotes the post-transpilation depth using IonQ's native gate set~\tablefootnote{https://docs.ionq.com/sdks/qiskit/native-gates-qiskit}.
\end{table}

\subsection{Scaling Analysis}
\label{sec:ionq-scaling}

\begin{figure}[htbp]
    \centering
    \includegraphics[width=0.99\linewidth]{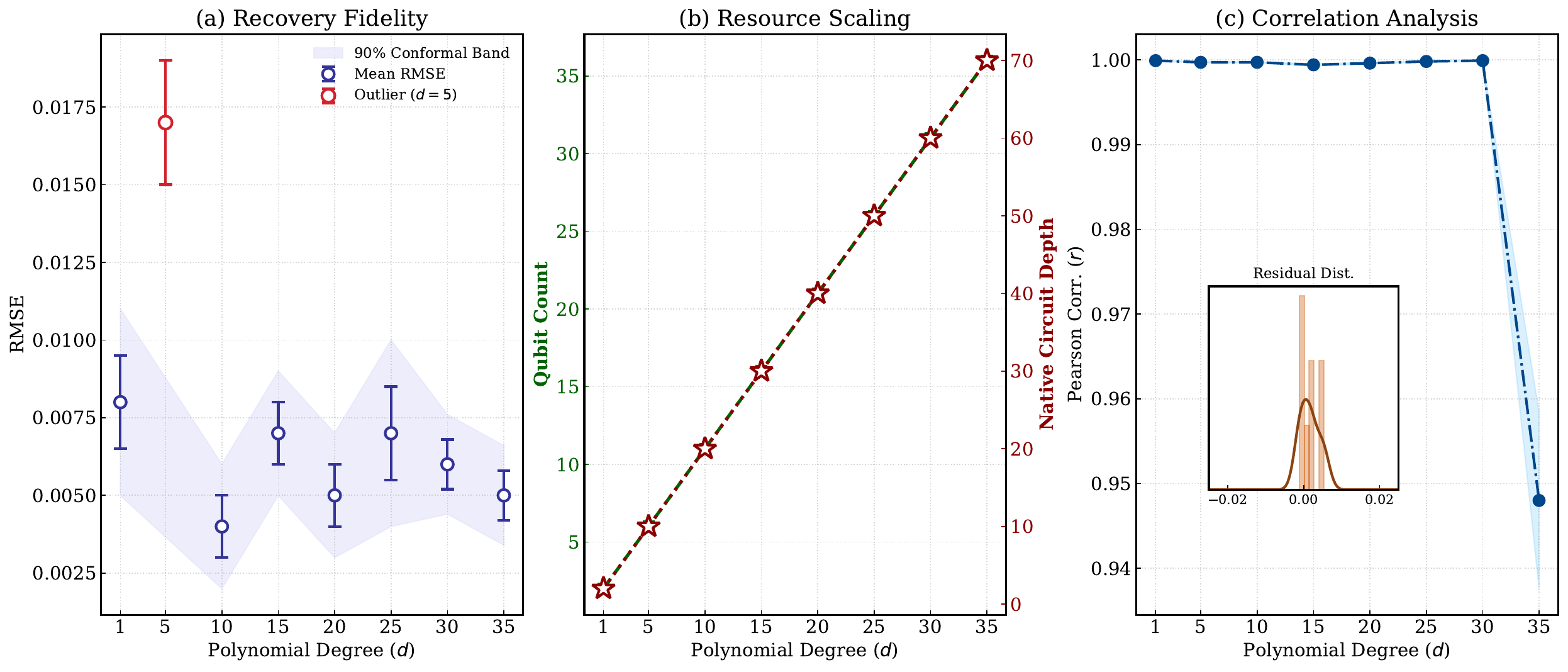}
    \caption{Scaling analysis on IonQ Forte-1: (a) RMSE, (b) correlation, (c) pass rate, and (d) circuit resources vs. degree.}
    \label{fig:ionq-scaling}
\end{figure}
\begin{figure}[htbp]
    \centering
    \includegraphics[width=0.99\linewidth]{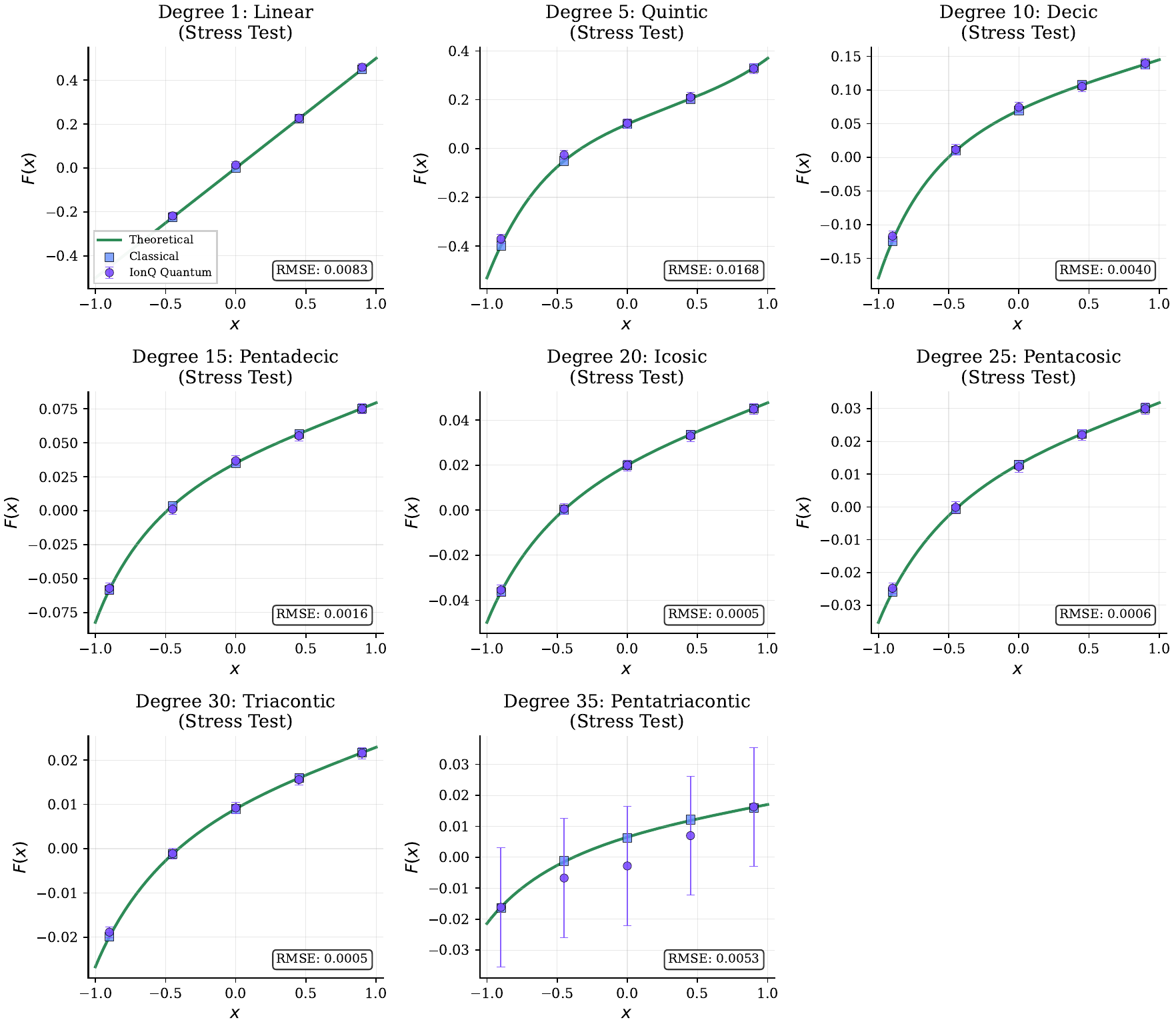}
    \caption{\nq\ performance on IonQ Forte-1 up to degree 35.}
    \label{fig:ionq-ibm-comparison}
\end{figure}

The non-monotonic RMSE trend (Figure~\ref{fig:ionq-scaling}a) reveals competing error sources. For degrees 5--30, the RMSE decreases as the normalized output magnitudes shrink, making the shot noise ($O(1/\sqrt{N})$) negligible. At degree 35, the gate errors increase the RMSE. The sharp correlation drop (Figure~\ref{fig:ionq-scaling}c) suggests a distinct decoherence boundary at a depth of 70. Resource scaling remains linear (Figure~\ref{fig:ionq-scaling}b), requiring $d+1$ qubits and $3d+1$ gates, thereby avoiding the exponential costs of variational methods.

Figure~\ref{fig:ionq-ibm-comparison} confirms that the trapped-ion coherence and connectivity enable stable execution up to 35 degrees. Our sparse sampling efficiently benchmarks the hardware quality without dense statistical characterization. For practical use, degree 10 polynomials are well within the NISQ capabilities, requiring only 11 qubits and 9 two-qubit gates. Future hardware could potentially support degrees 40--45 before coherence limits apply. Note that these results are specific to IonQ Forte-1; other architectures may vary.

\clearpage
\section{IBM Architecture Comparison: Nighthawk vs. Heron}
\label{sec:nighthawk_analysis}

We benchmarked the \nq\ approach on the IBM Quantum Nighthawk (`ibm\_miami`), which features a square-lattice topology with improved connectivity (218 couplers) relative to the heavy-hex Heron architecture.

\subsection{Circuit Compilation and Connectivity}

While Nighthawk's higher connectivity reduces the SWAP overheads for complex topologies, our polynomial circuits exhibit linear dependency chains ($x \to x^2 \to \dots \to x^d$) that map naturally to nearest-neighbor layouts. Transpilation analysis (\cref{tab:transpilation}) confirms that for degrees $d=3 \dots 6$, the standard Heron topology requires zero SWAP gates (transpiled = ideal count). Consequently, the topological advantages of Nighthawk yield no reduction in the gate count for this application, negating the expected improvements observed in general workloads.

\begin{table}[h]
\centering
\caption{Transpilation Analysis: 2-Qubit Gate Counts}
\label{tab:transpilation}
\begin{tabular}{ccccc}
\toprule
\textbf{Degree} & \textbf{Ideal 2Q Gates} & \textbf{Heron (Boston) 2Q} & \textbf{Nighthawk (Miami) 2Q} & \textbf{Reduction} \\
\midrule
3 & 2 & 2 & 2 & 0\% \\
4 & 3 & 3 & 3 & 0\% \\
5 & 3 & 3 & 3 & 0\% \\
6 & 4 & 4 & 4 & 0\% \\
\bottomrule
\end{tabular}
\end{table}

\subsection{Experimental Performance on Nighthawk}

Despite Nighthawk's currently higher gate error rates ($1-1.5\times$ vs. Heron), our experiments (10,000 shots) achieved high fidelity. As shown in Table~\ref{tab:miami_results}, the RMSE remained approximately $0.01$ with correlations exceeding $99.9\%$, demonstrating that the current coherence levels were sufficient.

\begin{table}[h]
\centering
\caption{Recovery Accuracy on IBM Miami (Nighthawk)}
\label{tab:miami_results}
\begin{tabular}{lccc}
\toprule
\textbf{Degree} & \textbf{Quantum RMSE} & \textbf{Correlation} & \textbf{Pass Rate} \\
\midrule
1 (Linear) & 0.0082 & 0.9997 & 100\% \\
2 (Quadratic) & 0.0093 & 0.9999 & 100\% \\
3 (Cubic) & 0.0143 & 0.9994 & 100\% \\
4 (Quartic) & 0.0104 & 0.9996 & 100\% \\
5 (Quintic) & 0.0132 & 0.9995 & 100\% \\
6 (Sextic) & 0.0128 & 0.9997 & 100\%  \\

\bottomrule
\end{tabular}
\end{table}

In summary, for linear topology algorithms, such as iterative polynomial multiplication, Nighthawk's connectivity offers no benefit over Heron, making Heron's lower gate error preferable. However, the successful high-accuracy execution on the Nighthawk confirms the robustness of \nq\ across diverse superconducting architectures.

\begin{figure*}[htbp]
    \centering
    \begin{subfigure}[b]{0.32\textwidth}
        \centering
        \includegraphics[width=\linewidth]{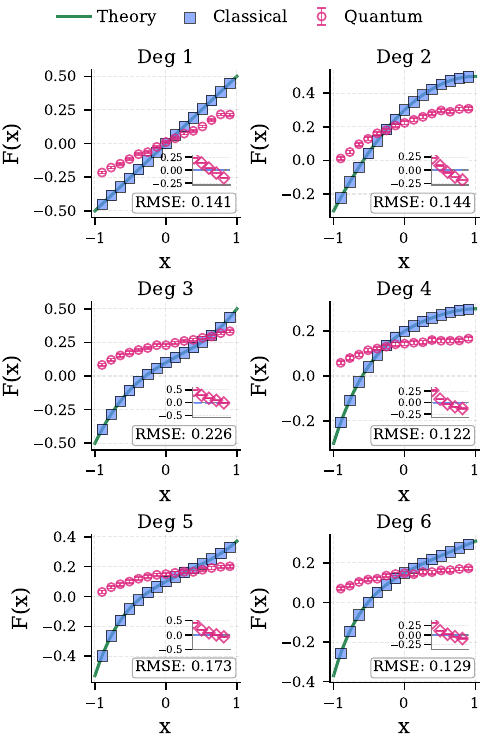}
    \end{subfigure}
    \hfill
    \begin{subfigure}[b]{0.245\textwidth}
        \centering
        \includegraphics[width=\linewidth]{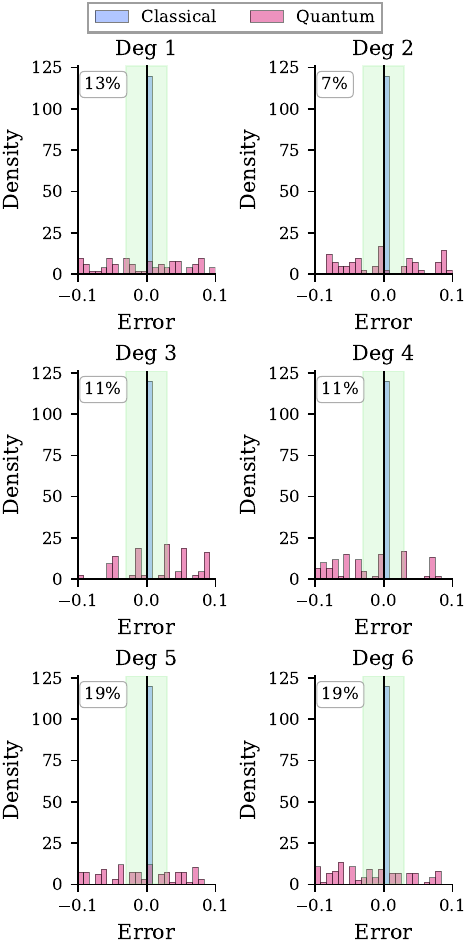}
    \end{subfigure}
    \hfill
    \begin{subfigure}[b]{0.32\textwidth}
        \centering
        \includegraphics[width=\linewidth]{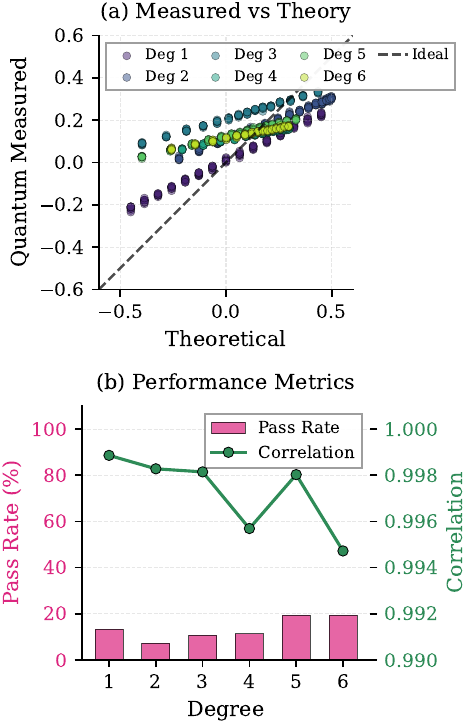}
    \end{subfigure}
    \caption{Analysis of polynomial recovery for degrees $d \in [1, 6]$ is presented across three comparative modalities. The left column illustrates the recovery grids featuring theoretical trajectories (green lines), classical predictions (squares), and quantum measurements (circles), supplemented by the residual error and RMSE. The middle column displays the error distribution histograms for the classical and quantum methodologies, identifying the PASS threshold region ($|\text{error}| < 0.03$) and associated pass rate annotations. The right column provides (a) a scatter plot correlating quantum measurements with theoretical benchmarks and (b) the evolution of pass rates and correlation coefficients as a function of polynomial degree.}
    \label{fig:cpt_error}
\end{figure*}

\section{Computation Error Analysis} 
\label{app:error}
We assess two methodologies for polynomial evaluation on quantum hardware: the direct method, which computes the polynomial classically and encodes the result into a single qubit, and the native method, which employs a closed-form transformation (see \cref{thm:coeff-to-weight}) to compute polynomial powers natively across multiple qubits. Our experiment indicates that incorporating tenfold more quantum arithmetic blocks transformed by the same classical weights is feasible. Consequently, we anticipate that our approach may underperform, given the limited T2 dephasing time of the IBM Heron3 QPU. The empirical findings demonstrate a trade-off between resource efficiency and computational accuracy, as shown in \cref{fig:cpt_error}.

The discrepancy between quantum measurements and classical predictions arises from the fundamental statistics of quantum measurements, as indicated by \cref{eq:shots_noise}. In the case of a single-qubit measurement with a finite number of shots, the measured probability adheres to a binomial distribution, which introduces a statistical uncertainty that affects the expectation value. The resulting distribution is symmetric and approximately normal, centered at zero, with a width consistent with shot-noise predictions, as illustrated on the left side of \cref{fig:cpt_error}. It is important to note that although \nq\ anticipates a tenfold increase in the coherence time to accurately reproduce results that align with the baseline, the algorithm still provides an exact correlation when compared to the direct encoding baseline. Additionally, it was observed that outlier errors were distributed across the encoded range [-1, 1] owing to the limitations imposed by trigonometric encoding on the rotation angle of each quantum gate. The degree-independent nature of this error further confirms that quantum arithmetic operations are exact.
We compare two evaluation methodologies: a 'direct' method encoding classically computed values into a single qubit and our 'native' method using recursive quantum arithmetic (\cref{thm:coeff-to-weight}). While the native method requires significantly more gates, testing the limits of the IBM Heron3 $T_2$ times, empirical results highlight a trade-off between implementation complexity and verification strictness (see \cref{fig:cpt_error}).

The discrepancies between the measurement and prediction stem fundamentally from finite sampling statistics (\cref{eq:shots_noise}). Single-qubit measurements follow a binomial distribution, yielding symmetric, normally distributed errors consistent with shot noise (\cref{fig:cpt_error}, left). Despite necessitating higher coherence, \nq\ maintains a high correlation with the direct baseline. Outliers across $[-1, 1]$ typically reflect the trigonometric encoding limits. Crucially, the error's independence from polynomial degree confirms the exactness of the quantum arithmetic operations

\end{document}